\documentclass[runningheads]{llncs}
\usepackage{amsmath}
\usepackage{amssymb}
\usepackage{appendix}
\usepackage{booktabs}
\usepackage{graphicx}
\usepackage{caption}
\usepackage{arydshln}
\usepackage{algorithm,algorithmic}
\usepackage{url}
\usepackage{hyperref}
\usepackage{lscape}
\usepackage{color}
\usepackage{multirow}
\usepackage{mathtools}
\usepackage{booktabs}
\usepackage{dsfont}
\usepackage{tikz}
\usepackage{enumerate}
\usepackage{comment}
\usepackage{enumitem}
\usepackage{subcaption}
\usepackage{subfloat}
\usepackage{indentfirst}
\usepackage[misc]{ifsym}
\usepackage{cancel}
\newtheorem{assumption}{Assumption}
\makeatletter
\DeclarePairedDelimiterX{\pmodx}[1]{(}{)}{{\operator@font mod}\mkern6mu#1}
\renewcommand{\pmod}{%
  \allowbreak
  \if@display\mkern18mu\else\mkern8mu\fi
  \pmodx
}

\hypersetup{pdfstartview={XYZ null null 1.80}}

\title{Generalized Implicit Factorization Problem}
\author{Yansong Feng \inst{1,2}\and Abderrahmane Nitaj \inst{3}$^{(\textrm{\Letter})}$\and Yanbin Pan \inst{1,2}$^{(\textrm{\Letter})}$}
\institute{ Key Laboratory of Mathematics Mechanization,   Academy of Mathematics and Systems Science, Chinese Academy of Sciences, Beijing, China \and		School of Mathematical Sciences, University of Chinese Academy of Sciences, Beijing, China \\\email{ysfeng2023@163.com}, \email{panyanbin@amss.ac.cn} \and Normandie Univ, UNICAEN, CNRS, LMNO, 14000 Caen, France \email{abderrahmane.nitaj@unicaen.fr}}
\begin{document}
\maketitle              
\begin{abstract}

The Implicit Factorization Problem (IFP) was first introduced by May and Ritzenhofen at PKC'09, which concerns the factorization of two RSA moduli $N_1=p_1q_1$ and $N_2=p_2q_2$, where $p_1$ and $p_2$ share a certain consecutive number of least significant bits. Since its introduction, many different variants  of IFP have been considered, such as the cases where $p_1$ and $p_2$ share most significant bits or middle bits at the same positions. In this paper, we consider a more generalized case of IFP,  in which the shared consecutive bits can be located  at \textit{any} positions in each prime, not necessarily required to be located at the same positions as before. We propose a lattice-based algorithm to solve this problem  under specific conditions, and also provide some experimental results to verify our analysis.
	
\end{abstract}

\section{Introduction}
In 1977, Rivest, Shamir, and Adleman proposed the famous RSA encryption scheme \cite{rivest1983cryptographic}, whose security is based on the hardness of factoring large integers. RSA is now a very popular scheme with many applications in industry for information security protection. Therefore, its security has been widely analyzed. Although it seems infeasible to break RSA with large modulus entirely with a classical computer now,
there still exist many vulnerable RSA instances.  For instance, small public key \cite{DBLP:conf/eurocrypt/Coppersmith96,DBLP:journals/joc/Coppersmith97} or small secret key \cite{DBLP:conf/eurocrypt/BonehD99} can lead to some attacks against RSA.
In addition, side-channel attacks pose a great threat to RSA \cite{DBLP:conf/ches/BauerJLPR14,DBLP:conf/uss/BrumleyB03,DBLP:conf/host/CarmonSW17}, targeting the decryption device to obtain more information about the private key.

It is well known that additional information on the private keys or the prime factors can help attack the  RSA scheme efficiently.  In 1997, Coppersmith \cite{DBLP:journals/joc/Coppersmith97,DBLP:phd/de/May2003} proposed an attack  that can factor the RSA modulus $N=pq$ in polynomial time if at least half of the most (or least) significant bits of $p$ are given. In 2013, by using Coppersmith's method, Bernstein et al. \cite{DBLP:conf/asiacrypt/BernsteinCCCHLS13} showed that an attacker can efficiently factor 184
distinct RSA keys generated by government-issued smart cards.

At PKC 2009, May and Ritzenhofen \cite{DBLP:conf/pkc/MayR09} introduced the Implicit Factorization Problem (IFP). It concerns the question of factoring two $n$-bit RSA moduli $N _1=p_1q_1$ and $N_2=p_2q_2$, given the implicit information that $p_1$ and $p_2$ share $\gamma n$ of their consecutive least significant bits, while $q_1$ and $q_2$ are $\alpha n$-bit. Using a two-dimensional lattice, May and Ritzenhofen obtained a heuristic result that this implicit information is sufficient to factor $N_1$ and $N_2$ with a lattice-based algorithm, provided that $\gamma n  > 2\alpha n +2$.

In a follow-up work at PKC 2010, Faug{\`{e}}re \textit{et al.} \cite{DBLP:conf/pkc/FaugereMR10} generalized the Implicit Factorization Problem to the case where the most significant bits (MSBs) or the middle bits of $p_1$ and $p_2$ are shared. Specifically, they established the bound of $\gamma n > 2\alpha n + 2$ for the case where the MSBs are shared, using a two-dimensional lattice. For the case where the middle bits of $p_1$ and $p_2$ are shared, Faug{\`{e}}re \textit{et al.} obtained a heuristic result that $q_1$ and $q_2$ could be found from a three-dimensional lattice if $\gamma n > 4\alpha n + 6$.

In 2011, Sarkar and Maitra \cite{DBLP:journals/tit/SarkarM11} further expanded the Implicit Factorization Problem by revealing the relations between the Approximate Common Divisor Problem (ACDP) and the Implicit Factorization Problem, and presented the bound of $\gamma > 2\alpha -\alpha^2$ for the following three cases.
\begin{enumerate}
	\item the primes $p_1$, $p_2$ share an amount of the least significant bits (LSBs);
	\item the primes $p_1$, $p_2$ share an amount of most significant bits (MSBs);
	\item the primes $p_1$, $p_2$ share both an amount of least significant bits and an amount of most significant bits.
\end{enumerate}

In 2016, Lu \textit{et al.} \cite{lu2016towards} presented a novel algorithm and improved the bounds to $\gamma > 2\alpha -2\alpha^2$ for all the above three cases of the Implicit Factorization Problem.
In 2015, Peng \textit{et al.} \cite{DBLP:conf/iwsec/PengHLHX15} revisited the Implicit Factorization Problem with shared middle bits and improved the bound of Faug{\`{e}}re \textit{et al.} \cite{DBLP:conf/pkc/FaugereMR10} up to $\gamma>4\alpha-3\alpha^2$. The bound was further enhanced by Wang \textit{et al.} \cite{DBLP:journals/chinaf/WangQLF18} in 2018 up to $\gamma>4\alpha-4\alpha\sqrt{\alpha}$.

It is worth noting that in the previous cases, the shared bits are located at the same position for the primes $p_1$ and $p_2$.

In this paper, we present a more generalized case of the Implicit Factorization Problem that allows for arbitrary consecutive shared locations, rather than requiring them to be identical in the primes, as in previous research. More precisely, we propose the Generalized Implicit Factorization Problem (GIFP), which concerns the factorization of two $n$-bit RSA moduli $N_1=p_1q_1$ and $N_2=p_2q_2$ when $p_1$ and $p_2$ share $\gamma n$ consecutive bits, where the shared bits are not necessarily required to be located at the same positions. See Fig. \ref{fig:1} for an example, where the starting  positions for the shared bits in $p_1$ and $p_2$ may be different.

\begin{figure}[htbp]
	\subfloat[$p_1$]
	{
		\begin{minipage}{0.5\textwidth}
			\centering
			\includegraphics[scale=0.8]{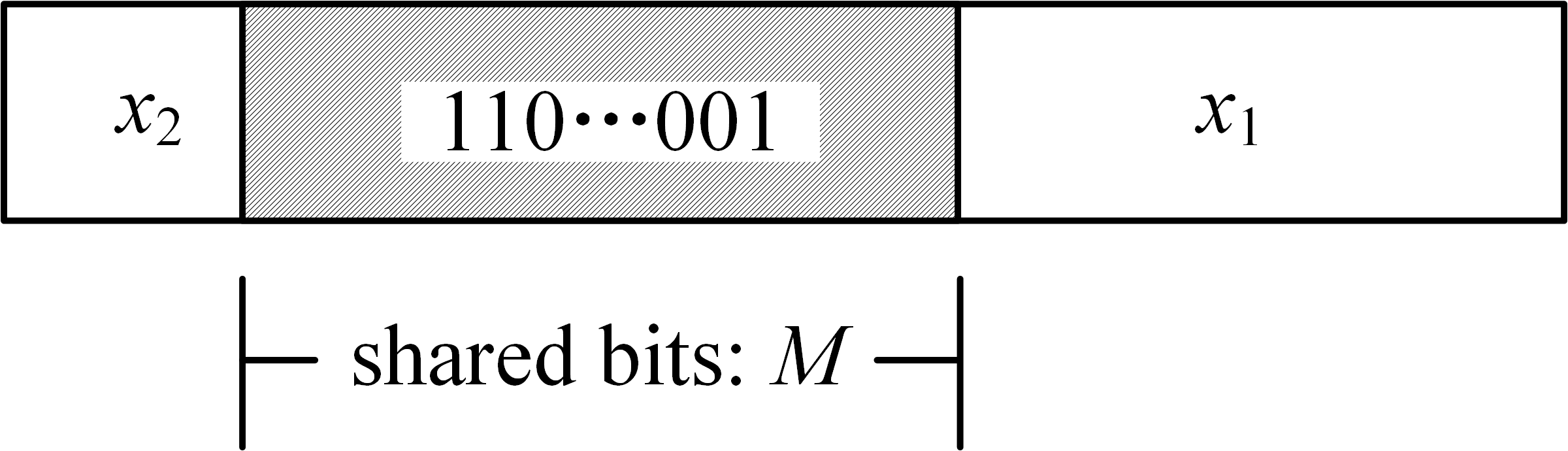}
		\end{minipage}
	}
	\subfloat[$p_2$]
	{
		\begin{minipage}{0.5\textwidth}
			\centering
			\includegraphics[scale=0.8]{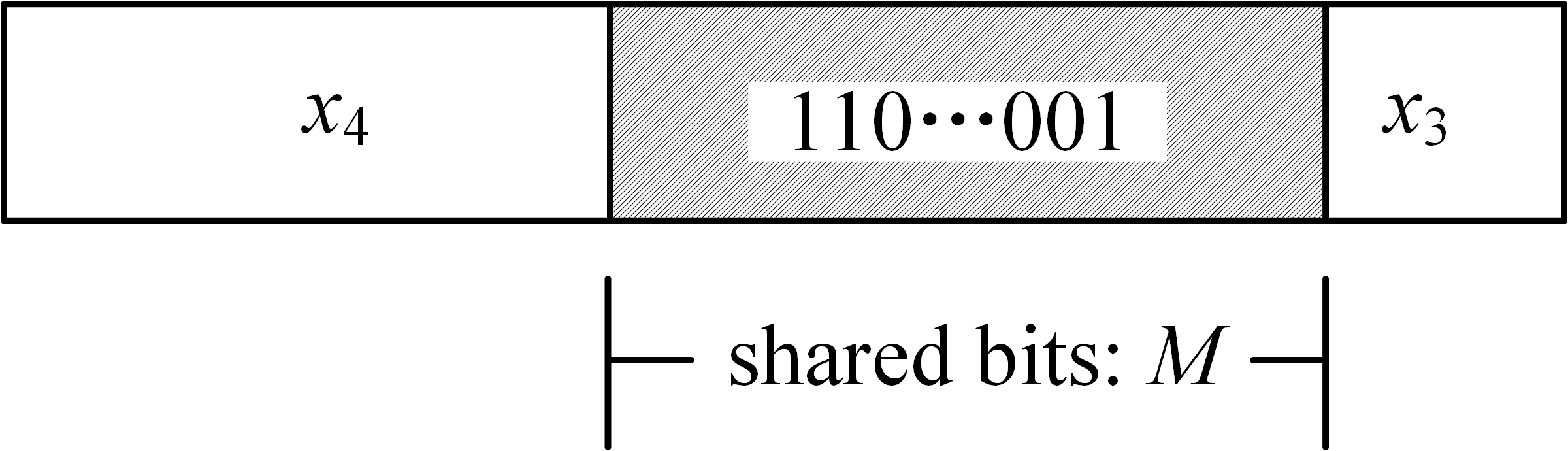}
		\end{minipage}
	}
	\caption{Shared bits $M$ for $p_1$ and $p_2$} %
 \label{fig:1}
\end{figure}
We transform the GIFP into the Approximate Common Divisor Problem and then, employ Coppersmith's method with some optimization strategy, we propose a polynomial time algorithm to solve it when $\gamma>4\alpha(1-\sqrt{\alpha})$.

In Table  \ref{Table0}, we present a comparison of our new bound on $\gamma$ with the known former bounds obtained by various methods to solve the Implicit Factorization Problem.

\begin{table*}
\small{
\centering
\begin{tabular}{cccccc}
   \toprule
& LSBs & MSBs & both LSBs-MSBs &Middle bits &General \\
   \midrule
		May, Ritzenhofen \cite{DBLP:conf/pkc/MayR09} &$2\alpha$&-&-&-&-\\
		Faug{\`{e}}re, \textit{et al.} \cite{DBLP:conf/pkc/FaugereMR10}&$2\alpha$&-&-&$4\alpha$&-\\
		Sarkar, Maitra \cite{DBLP:journals/tit/SarkarM11} & $2\alpha-\alpha^2$ &$2\alpha-\alpha^2$ & $2\alpha-\alpha^2$&-&-\\
		Lu, \textit{et al.} \cite{lu2016towards} & $2\alpha-2\alpha^2$ &$2\alpha-2\alpha^2$ & $2\alpha-2\alpha^2$&-&-\\
		Peng, \textit{et al.}\cite{DBLP:conf/iwsec/PengHLHX15} &-&-&-&$4\alpha-3\alpha^2$&-\\
		Wang, \textit{et al.}\cite{DBLP:journals/chinaf/WangQLF18} &-&-&-&$4\alpha(1-\sqrt{\alpha})$&-\\
		\bf{This work}&-&-&-&-&$4\alpha(1-\sqrt{\alpha})$\\
   \bottomrule
\end{tabular}
\caption{Asymptotic lower bound of $\gamma$ in the Implicit Factorization Problem for $n$-bit $N_1=p_1q_2$ and $N_2=p_2q_2$ where the number of shared bits is $\gamma n$,  $q_1$ and $q_2$ are $\alpha n$-bit.}
\label{Table0}
}
\end{table*}

It can be seen in Table \ref{Table0} that the bounds for the Implicit Factorization Problem for sharing middle bits are inferior to those of other variants. This is because the unshared bits in the Implicit Factorization Problem for LSBs or MSBs or both LSBs and MSBs  are continuous, and only one variable is necessary  to represent the unshared bits while at least two variables are needed to represent the unshared bits in the Implicit Factorization Problem sharing middle bits or GIFP. In addition, our bound for GIFP is identical to the variant of IFP sharing the middle bits located in the same position. However, it is obvious that the GIFP relaxes the constraints for the positions of the shared bits.

Therefore, with the same bound for the number of shared bits as in the IFP sharing middle bits at the same position, we show that the Implicit Factorization Problem can still be solved efficiently when the positions for the sharing bits are located differently.

There are still open problems, and the most important one is: can we improve our bound $4\alpha\left(1-\sqrt{\alpha}\right)$ for GIFP to $2\alpha-2\alpha^2$ or even better? A positive answer seems not easy since the bound for GIFP directly yields a bound for any known variant of IFP. Improving the bound for GIFP to the one better than $4\alpha\left(1-\sqrt{\alpha}\right)$  means that we can improve the bound for the variant of IFP sharing the middle bits located in the same position, and improving the bound for GIFP to the one better than $2\alpha-2\alpha^2$  means that we can improve the bound  for any known variant of IFP.

\subsubsection{Roadmap}
Our paper is structured as follows. Section 2 presents some required background for our approaches. In Section 3, we present our analysis of the Generalized Implicit Factorization Problem, which constitutes our main result. Section 4 details the experimental results used to validate our analysis. Finally, we provide a brief conclusion in Section 5. The source code is available at:\begin{center}
  \textcolor{blue}{\url{https://github.com/fffmath/gifp}}.
\end{center}

\section{Notations and Preliminaries}
\subsubsection{Notations}
Let $\mathbb{Z}$ denote the ring of integers, i.e., the set of all integers. We use lowercase bold letters (e.g., $\mathbf{v}$) for vectors and uppercase bold letters (e.g., $\mathbf{A}$) for matrices. The notation $\binom{n}{m}$ represents the number of ways to select $m$ items out of $n$ items, which is defined as $\frac{n!}{m!(n-m)!}$. If $m>n$, we set $\binom{n}{m}=0$.

\subsection{Lattices, SVP, and LLL}
Let $m\geq 2$ be an integer. A lattice is a discrete additive subgroup of $\mathbb{R}^m$. A more explicit definition is presented as follows.
\begin{definition}[Lattice]
Let  $\mathbf{v_1},\mathbf{v_2},\dots,\mathbf{v_n}\in \mathbb{R}^m$ be $n$ linearly independent vectors with $n\leq m$. The lattice $\mathcal{L}$ spanned by $\left\{\mathbf{v_1},\mathbf{v_2},\dots,\mathbf{v_n} \right\}$ is the set  of all integer linear combinations of $\left\{\mathbf{v_1},\mathbf{v_2},\dots,\mathbf{v_n} \right\}$, i.e.,
	$$
	\mathcal{L}=\left\{\mathbf{v}\in \mathbb{R}^m\ |\ \mathbf{v}=\sum_{i=1}^n a_i\mathbf{v_i}, a_i\in \mathbb{Z}\right\}.
	$$
\end{definition}

The integer $n$ denotes the rank of the lattice $\mathcal{L}$, while $m$ represents its dimension. The lattice $\mathcal{L}$ is said to be full rank if $n=m$. We use the matrix $\mathbf{B}\in \mathbb{R}^{n\times m}$, where each vector $\mathbf{v_i}$ contributes a row to $\mathbf{B}$. The determinant of $\mathcal{L}$ is defined as $\det(\mathcal{L})=\sqrt{\det\left(\mathbf{B}\mathbf{B}^t\right)}$, where $\mathbf{B}^t$ is the transpose of $\mathbf{B}$. If $\mathcal{L}$ is full rank, this reduces to $\det(\mathcal{L})=\left|\det\left(\mathbf{B}\right)\right|$.



The Shortest Vector Problem (SVP) is one of the famous computational problems in lattices.
\begin{definition}[Shortest Vector Problem (SVP)]
Given a lattice $\mathcal{L}$, the Shortest Vector Problem (SVP) asks to find a non-zero lattice vector $\mathbf{v}\in \mathcal{L}$ of minimum Euclidean norm, i.e., find $\mathbf{v}\in \mathcal{L}\backslash \{\mathbf{0}\}$ such that $\|\mathbf{v}\|\leq \|\mathbf{w}\|$ for all non-zero $\mathbf{w}\in \mathcal{L}$.
\end{definition}

Although SVP is NP-hard  under randomized reductions \cite{Ajtai1998TheSV}, there exist algorithms that can find a relatively short vector, instead of the exactly shortest vector,  in polynomial time, such as the famous LLL algorithm proposed by Lenstra, Lenstra, and Lovasz \cite{lenstra1982factoring}  in 1982. The following result is useful for our analysis\cite{DBLP:phd/de/May2003}.

\begin{theorem}[LLL Algorithm]\label{LLL}
Given an $n$-dimensional lattice $\mathcal{L}$, we can find an \text{LLL}-reduced basis $\left\{\mathbf{v_1},\mathbf{v_2},\dots,\mathbf{v_n}\right\}$ of $\mathcal{L}$ in polynomial time, which satisfies
	$$
	\Vert \mathbf{v_i}\Vert\leq 2^{\frac{n(n-1)}{4(n+1-i)}}\det(\mathcal{L})^{\frac{1}{n+1-i}},\quad \text{for}\quad i=1,\dots,n .
	$$
\end{theorem}

 Theorem \ref{LLL} presents the upper bounds for the norm of the $i$-th vector in the LLL-basis using the determinant of the lattice.

\subsection{Coppersmith's method}
In 1996, Coppersmith \cite{DBLP:journals/joc/Coppersmith97,DBLP:phd/de/May2003} proposed a lattice-based method for finding small solutions of univariate modular polynomial equations modulo a positive integer $M$, and another lattice-based method for finding the small roots of bivariate polynomial equations. The methods are based on finding short vectors in a lattice. We briefly sketch the idea below. More details can be found in \cite{DBLP:phd/de/May2003}.

Let $M$ be a positive integer, and $f(x_1,\dots,x_k)$ be a polynomial with integer coefficients. Suppose we want to find a small solution $(y_1,\dots,y_k)$ of the modular equation $f(x_1,\dots,x_k) \equiv 0 \pmod{M}$ with the bounds $y_i<X_i$ for $i=1,\ldots, k$.

The first step is to construct a set $G$ of $k$-variate polynomial equations such that, for each $g_i\in G$ with $i=1,\ldots, k$, we have $g_i(y_1,\dots,y_k)\equiv 0 \pmod{M}$. Then we use the coefficient vectors of $g_i(x_1X_1,\dots,x_kX_k)$, $i=1,\ldots, k$, to construct a $k$-dimensional lattice $\mathcal{L}$. Applying the LLL algorithm to $\mathcal{L}$, we get a new set $H$ of $k$ polynomial equations $h_i(x_1,\dots,x_k)$, $i=1,\ldots,k$, with integer coefficients such that $h_i(y_1,\dots,y_k)\equiv 0 \pmod{M}$. The following result shows that one can get $h_i(y_1,\dots,y_k)=0$ over the integers in some cases, where for $h(x_1,\dots,x_k)=\sum_{i_1\ldots i_k} a_{i_1\ldots i_k}x_1^{i_1}\cdots x_1^{i_k}$, the Euclidean norm  is defined   by
$
\left\Vert h(x_1,\dots,x_k)\right\Vert=\sqrt{\sum_{i_1\ldots i_k}a_{i_1\ldots i_k}^2}
$.
\begin{theorem}[Howgrave-Graham \cite{howgrave1997finding}]\label{HG}
 Let $h(x_1,\dots,x_k)\in \mathbb{Z}[x_1,\ldots ,x_k]$ be a polynomial with at most $\omega$ monomials. Let $M$ be a positive integer. If there exist $k$ integers $(y_1,\dots,y_k)$ satisfying the following two conditions:
\begin{enumerate}
\item $h(y_1,\dots,y_k)\equiv 0 \pmod{M}$, and there exist $k$ positive integers $X_1,\dots,X_k$ such that $\left|y_1\right| \leq X_1,\ldots ,\left|y_k\right| \leq X_k$,
\vspace{0.2cm}
\item $\left\Vert h(x_1X_1,\dots,x_kX_k)\right\Vert<\frac{M}{\sqrt{\omega}}$,
\end{enumerate}
then $h(y_1,\dots,y_k)=0$ holds over the integers.
\end{theorem}

From Theorem~\ref{LLL},   we can obtain the vectors $\mathbf{v_1},\mathbf{v_2},\dots,\mathbf{v_k}$ in the LLL reduced basis of $\mathcal{L}$. This yields $k$ integer polynomials $h_1(x_1,\dots,x_k)$, $\dots$, $h_k(x_1,\dots,x_k)$, all of which share the desired solution $(y_1,\dots,y_k)$, that is $h_i(y_1,\dots,y_k)\equiv 0\pmod{M}$ for $i=1,\ldots, k$.

To combine Theorem~\ref{LLL} and Theorem~\ref{HG}, for $i=k$, we set
$$
2^{\frac{n(n-1)}{4(n+1-i)}}\det(\mathcal{L})^{\frac{1}{n+1-i}}<\frac{M}{\sqrt{\dim(\mathcal{L})}}.
$$

Ultimately, the attainment of the desired root hinges upon effectively resolving the system of integer polynomials using either the resultant method or the Gröbner basis approach. However, in order for a Gröbner basis
computation to find the common root, the following heuristic assumption needs to hold.
\begin{assumption}\label{assu::solution}
The $k$ polynomials $h_i(x_1,\cdots,x_k)$, $i=1,\cdots, k$, that are derived from the reduced basis of the lattice in the Coppersmith
method are algebraically independent. Equivalently,
the common root of the polynomials $h_i(x_1,\cdots,x_k)$
can be found by computing the resultant or computing the Gr\"{o}bner basis.
\end{assumption}
Assumption \ref{assu::solution} is often used in connection with Coppersmith's method in the multivariate scenario
 \cite{DBLP:conf/eurocrypt/BonehD99,DBLP:phd/de/May2003,DBLP:journals/tit/SarkarM11,lu2016towards,DBLP:journals/chinaf/WangQLF18}.
Since our attack in Section \ref{attack} relies on  Assumption \ref{assu::solution}, it is heuristic. However, our experiments in Section \ref{exp}
justify the validity of our attack and
show that  Assumption \ref{assu::solution} perfectly holds true.

\section{Generalized Implicit Factorization Problem}\label{attack}
This section presents our analysis of the Generalized Implicit Factorization Problem (GIFP) in which $p_1$ and $p_2$ share an amount of consecutive bits at different positions.

\subsection{Description of GIFP}
This section proposes the Generalized Implicit Factorization Problem (GIFP), which concerns the factorization of two $n$-bit RSA moduli, $N_1=p_1q_1$ and $N_2=p_2q_2$, under the implicit hint that the primes $p_1$ and $p_2$ share a specific number, $\gamma n$, of consecutive bits. 

In contrast to previous studies \cite{DBLP:conf/pkc/FaugereMR10,lu2016towards,DBLP:conf/pkc/MayR09,DBLP:journals/amco/SarkarM09,DBLP:journals/amco/SarkarM10,DBLP:journals/chinaf/WangQLF18}, where the shared bits were assumed to be located at the same positions in $p_1$ and $p_2$, the proposed GIFP considers a more general case where the shared bits can be situated at arbitrary positions.

\begin{definition}[GIFP($n, \alpha, \gamma$)]
Given two $n$-bit RSA moduli $N_1=p_1q_1$ and $N_2=p_2q_2$, where  $q_1$ and  $q_2$ are $\alpha n$-bit, assume that $p_1$ and $p_2$ share $\gamma n$ consecutive bits, where the shared bits may be located in different positions of $p_1$ and $p_2$. The Generalized Implicit Factorization Problem (GIFP) asks to  factor $N_1$ and $N_2$.
\end{definition}

The introduction of GIFP expands the scope of the Implicit Factorization Problem and presents a more realistic and challenging scenario that can arise in practical applications. In real-world settings, it is more probable to encounter situations where the shared location of bits differs between primes. Therefore, it is essential to develop algorithms and analysis that can handle such cases where the shared bits are situated at different positions. By considering the Generalized Implicit Factorization Problem (GIFP), we need to avoid situations where the system that creates RSA keys lack entropy.

\subsection{Algorithm for GIFP}
We will show our analysis of the GIFP in this subsection. The main idea is also to relate the Approximate Common Divisor Problem (ACDP) to the Implicit Factorization Problem.

\begin{theorem}
Under Assumption \ref{assu::solution}, GIFP($n, \alpha, \gamma$) can be solved in polynomial time when
$$
\gamma>4\alpha\left(1-\sqrt{\alpha}\right),
$$
provided that $\alpha+\gamma\leq 1$.

\end{theorem}
\begin{proof}  Without loss of generality, we can assume that the starting and ending positions of the shared bits are known.  When these positions are unknown, we can simply  traverse the possible starting positions of the shared bits, which will just scale the time complexity  for the case that we know the position by a factor  $\mathcal{O}(n^2)$.

Hence, we suppose that $p_1$ shares $\gamma n$-bits from the $\beta_1 n$-th bit to $(\beta_1+\gamma) n$-th bit, and  $p_2$ shares bits from $\beta_2 n$-th bit to $(\beta_2+\gamma) n$-th bit, where $\beta_1$ and $\beta_2$ are known with $\beta_1\leq \beta_2$ (see Fig. \ref{fig:1} ). Then we can write
$$
p_1=x_1+M_02^{\beta_1 n}+x_22^{(\beta_1+\gamma) n},\quad p_2=x_3+M_02^{\beta_2 n}+x_42^{(\beta_2+\gamma) n},
$$
with $M_0<2^{\gamma n}$, $x_1<2^{\beta_1 n}$, $x_2<2^{(\beta-\beta_1) n}$, $x_3<2^{\beta_2 n}$, $x_4<2^{(\beta-\beta_2) n}$ where $\beta=1-\alpha-\gamma$. From this, we deduce
\begin{align*}
2^{(\beta_2-\beta_1) n}p_1&=x_12^{(\beta_2-\beta_1) n}+M_02^{\beta_2 n}+x_22^{(\beta_2+\gamma) n}\\
&=x_12^{(\beta_2-\beta_1) n}+(p_2-x_3-x_42^{(\beta_2+\gamma) n})+x_22^{(\beta_2+\gamma) n}\\
&=p_2+(x_12^{(\beta_2-\beta_1) n}-x_3)+(x_2-x_4)2^{(\beta_2+\gamma) n}.
\end{align*}
Then, multiplying by $q_2$, we get
$$
N_2+(x_12^{(\beta_2-\beta_1) n}-x_3)q_2+(x_2-x_4)q_22^{(\beta_2+\gamma) n}=2^{(\beta_2-\beta_1) n}p_1q_2.
$$
Next, we define the polynomial $$f(x,y,z)=xz+2^{(\beta_2+\gamma) n}yz+N_2,$$ which shows that $(x_12^{(\beta_2-\beta_1) n}-x_3,x_2-x_4,q_2)$ is a solutions of $$f(x,y,z)\equiv 0 \pmod{2^{(\beta_2-\beta_1) n}p_1}.$$
Let $m$ and $t$ be  integers to be optimized later with $0\leq t \leq m$. To apply Coppersmith's method, we consider a family of polynomials $g_{i,j}(x,y,z)$ for $0\leq i\leq m$ and $0\leq j\leq m-i$:
$$
g_{i,j}(x,y,z)=(yz)^{j}f(x,y,z)^i\left(2^{(\beta_2-\beta_1) n}\right)^{m-i}N_1^{\max(t-i,0)}.
$$
These polynomials satisfy
\begin{align*}
g_{i,j}&\left(x_1 2^{(\beta_2-\beta_1) n}-x_3,x_2-x_4,q_2\right)\\
=&\ (x_2-x_4)^{j}q_2^{j}\left(2^{(\beta_2-\beta_1) n}p_1q_2\right)^i\left(2^{(\beta_2-\beta_1) n}\right)^{m-i}N_1^{\max(t-i,0)}\\
=&\ (x_2-x_4)^{j}q_2^{j+i}q_1^{\max(t-i,0)}\left(2^{(\beta_2-\beta_1) n}\right)^mp_1^{\max(t-i,0)+i}\\
\equiv&\  0\pmod*{\left(2^{(\beta_2-\beta_1) n}\right)^mp_1^{t}}.
\end{align*}
On the other hand, we have
\begin{align*}
\left\vert x_12^{(\beta_2-\beta_1) n}-x_3\right\vert
&\leq \max\left(x_12^{(\beta_2-\beta_1) n},x_3\right)\\
&\leq \max\left(2^{\beta_1 n}2^{(\beta_2-\beta_1) n},2^{\beta_1 n}\right)\\
&=2^{\beta_2 n},
\end{align*}
and
$$
\vert x_2-x_4\vert\leq \max(x_2,x_4)=2^{(\beta-\beta_1)n}.
$$
Also, we have $ q_2 =2^{\alpha n}$.
We then set
$$
X=2^{\beta_2 n},\ Y=2^{(\beta-\beta_1) n},\ Z=2^{\alpha n}.
$$
To reduce the determinant of the lattice, we introduce a new variable $w$ for $p_2$, and multiply the polynomials $g_{i,j}(x,y,z)$ by a power $w^s$ for some $s$ that will be optimized later. Similar to $t$, we also require $0\leq s \leq m$

Note that we can replace $zw$ in $g_{i,j}(x,y,z)w^{s}$ by $N_2$. We want to eliminate this multiple. Since $\gcd(N_2, 2N_1)=1$, there exists an inverse of $N_2$, denoted as $N_2^{-1}$, such that $N_2N_2^{-1}\equiv 1 \pmod*{\left(2^{(\beta_2-\beta_1) n}\right)^mN_1^{t}}$. We then eliminate $(zw)^{\min(s,i+j)}$ from the original polynomial by multiplying it by $N_2^{-\min(s,i+j)}\pmod*{\left(2^{(\beta_2-\beta_1) n}\right)^mN_1^{t}}$, while ensuring that the resulting polynomial evaluation is still a multiple of $ \left(2^{(\beta_2-\beta_1) n}\right)^mp_1^{t}$. By selecting the appropriate parameter $s$, we aim to reduce the determinant of the lattice. To this end, we consider a new family of polynomials $G_{i,j}(x,y,z,w)$ for $0\leq i\leq m$ and $0\leq j\leq m-i$:
$$
G_{i,j}(x,y,z,w)=(yz)^{j}w^sf(x,y,z)^i\left(2^{(\beta_2-\beta_1) n}\right)^{m-i}N_1^{\max(t-i,0)}N_2^{-\min(s,i+j)},
$$
where $N_2^{-\min(s,i+j)}$ is computed modulo $\left(2^{(\beta_2-\beta_1) n}\right)^mN_1^{t}$, and each term $zw$ is replaced by $N_2$. For example, suppose $s\geq 1$, then 
$$G_{0,1}(x,y,z,w)=yw^{s-1}N_2\left(2^{(\beta_2-\beta_1) n}\right)^{m}N_1^{t}N_2^{-1}.$$
Next, consider the lattice $\mathcal{L}$ spanned by the matrix $\mathbf{B}$ whose rows are the coefficients of the polynomials $G_{i,j}(Xx,Yy,Zz,Ww)$ where, for  $0\leq i\leq m$, $0\leq j \leq m-i$,
The rows are ordered following the rule that $G_{i,j}\prec G_{i',j'}$ if $i<i'$ or if $i=i'$ and $j<j'$. 
The columns are ordered following the monomials so that $x^iy^jz^{i+j-\min(s,i+j)}w^{s-\min(s,i+j)}\prec x^{i'}y^{j'}z^{i'+j'-\min(s,i'+j')}w^{s-\min(s,i'+j')}$ if $i<i'$ or if $i=i'$ and $j<j'$.
Table \ref{TableGIFP0} presents a matrix $\mathbf{B}$ with $m=3$, $s=2$, $t=2$ where $\ast$ represents a nonzero term.

\begin{table}[h]
\begin{eqnarray*}
\tiny{
 \begin{array}{||c||c|c|c|c|c|c|c|c|c|c||}
\hline
G_{i,j}&w^2 &yw& y^2&y^3z &xw&xy&xy^2z&x^2&x^2yz&x^3z\\
  \hline\hline
G_{0, 0}&W^2M^3N_1^2&0&0&0&0&0&0&0&0&0\\
G_{0, 1}&0&YWM^3N_1^2&0&0&0&0&0&0&0&0\\
G_{0, 2}&0&0&Y^2M^3N_1^2&0&0&0&0&0& 0&0\\
G_{0, 3}&0&0&0&Y^3ZM^3N_1^2&0&0&0&0&0&0\\
G_{1, 0}&\ast&\ast&0&0&XWM^2N_1&0&0&0&0& 0\\
G_{1, 1}&0&\ast&\ast&0&0&XYM^2N_1&0&0&0&0\\
G_{1, 2}&0&0&\ast&\ast& 0&0&XY^2ZM^2N_1&0&0&0\\
G_{2,0}&\ast&\ast&\ast&0&\ast&\ast& 0&X^2M&0&0\\
G_{2,1}&0&\ast&\ast&\ast&0&\ast&\ast&0&X^2YZM& 0\\
G_{3,0}&\ast&\ast&\ast&\ast&\ast&\ast&\ast& \ast&\ast&X^3Z\\
\hline
  \end{array}
}
\end{eqnarray*}
\caption{The  matrix of the lattice  with $m=3$, $s=2$, $t=2$  and $M=2^{(\beta_2-\beta_1) n}$.}\label{TableGIFP0}
\end{table}

By construction, the square matrix $B$ is left triangular. Hence, the dimension of the lattice is
$$
\omega=\sum_{i=0}^{m}\sum_{j=0}^{m-i}1=\sum_{i=0}^{m}(m-i+1)=\frac{1}{2}(m+1)(m+2)
$$
and its determinant is
$$
\det(B)=\det(\mathcal{L})=X^{e_X}Y^{e_Y}Z^{e_Z}W^{e_W}2^{(\beta_2-\beta_1) ne_M}N_1^{e_N},
$$
with
\begin{align*}
e_X&=\sum_{i=0}^m\sum_{j=0}^{m-i}i=\frac{1}{6}m(m+1)(m+2),\\
e_Y&=\sum_{i=0}^m\sum_{j=0}^{m-i}j=\frac{1}{6}m(m+1)(m+2),\\
e_Z&=\sum_{i=0}^m\sum_{j=0}^{m-i}(i+j-\min(s,i+j))\\
&=\frac{1}{3}m(m+1)(m+2)+\frac{1}{6}s(s+1)(s+2)-\frac{1}{2}s(m+1)(m+2),\\
e_W&=\sum_{i=0}^m\sum_{j=0}^{m-i}(s-\min(s,i+j))=\frac{1}{6}s(s+1)(s+2),
\\
e_N&=\sum_{i=0}^t\sum_{j=0}^{m-i}(t-i)=\frac{1}{6}t(t+1)(3m-t+4),\\
e_M&=\sum_{i=0}^m\sum_{j=0}^{m-i}(m-i)=\frac{1}{3}m(m+1)(m+2).
\end{align*}
The former results are detailed in Appendix \ref{apdproof}. To combine Theorem~\ref{LLL} and Theorem~\ref{HG}, we set
$$
2^\frac{\omega(\omega-1)}{4(\omega+1-i)}
\det(\mathcal{L})^\frac{1}{\omega+1-i}<\frac{\left(2^{(\beta_2-\beta_1) n}\right)^mp_1^t}{\sqrt{\omega}},
$$
with $i=2$.
Then
$$
\det(\mathcal{L})<\frac{1}{2^\frac{\omega-1}{4}\sqrt{\omega}}\left(2^{(\beta_2-\beta_1) n}\right)^{\omega m}p_1^{t\omega},
$$
and
\begin{eqnarray}
X^{e_X}Y^{e_Y}Z^{e_Z}W^{e_W}2^{(\beta_2-\beta_1) ne_M}N_1^{e_N}<\frac{1}{2^\frac{\omega-1}{4}\sqrt{\omega}}\left(2^{(\beta_2-\beta_1) n}\right)^{\omega m}p_1^{t\omega}.\label{detGIFP}
\end{eqnarray}
Next, we set $s=\sigma m$ with $0\leq\sigma\leq 1$, $t=\tau m$ with $0\leq\tau\leq 1$, and we use $N\approx 2^n$, $p_1\approx 2^{(1-\alpha)n}$, $X=2^{\beta_2 n}$, $Y=2^{(\beta-\beta_1) n}$, $Z=2^{\alpha n}$, $W=2^{(1-\alpha)n}$ and the most significant parts of $e_X$, $e_Y$, $e_Z$, $e_W$, $e_N$, $e_M$ as
\begin{align*}
e_X&=\frac{1}{6}m^3+o\left(m^3\right),\\
e_Y&=\frac{1}{6}m^3+o\left(m^3\right),\\
e_Z&=\frac{1}{3}m^3+\frac{1}{6}\sigma^3m^3-\frac{1}{2}\sigma m^3+o\left(m^3\right),\\
e_W&=\frac{1}{6}\sigma^3m^3+o\left(m^3\right),\\
e_N&=\frac{1}{6}\tau^2(3-\tau)m^3+o\left(m^3\right),\\
e_M&=\frac{1}{3}m^3+o\left(m^3\right).
\end{align*}
Similarly, we use
$$
m\omega=\frac{1}{2}m^3+o\left(m^3\right).
$$
Then, after taking logarithms,  dividing by $nm^3$,  and neglecting the very small terms, i.e., $o\left(m^3\right)$, the inequality (\ref{detGIFP}) implies
\begin{align*}
&\frac{1}{6}\beta_2+\frac{1}{6}(\beta-\beta_1)+\alpha(\frac{1}{3}+\frac{1}{6}\sigma^3-\frac{1}{2}\sigma)+\frac{1}{6}\sigma^3(1-\alpha)+\frac{1}{3}(\beta_2-\beta_1)+\frac{1}{6}\tau^2(3-\tau)\\
&<\ \frac{1}{2}(\beta_2-\beta_1)+\frac{1}{2}(1-\alpha)\tau.
\end{align*}
Using $\beta=1-\alpha-\gamma$, the former inequality is equivalent to
$$
\tau^2(3-\tau)-3(1-\alpha)\tau+\sigma^3-3\alpha\sigma+1-\gamma+\alpha<0.
$$
The left side is optimized for $\tau_0=1-\sqrt{\alpha}$ and $\sigma_0=\sqrt{\alpha}$, which gives
$$
3\alpha-2\alpha\sqrt{\alpha}-1-2\alpha\sqrt{\alpha}+1+\alpha-\gamma<0,
$$
and finally
$$
\gamma>4\alpha\left(1-\sqrt{\alpha}\right).
$$
By Assumption \ref{assu::solution}, we can get $(x_0,y_0,z_0)=(x_12^{(\beta_2-\beta_1) n}-x_3,x_2-x_4,q_2)$, so we have $q_2=z_0$, and we calculate $$
p_2=\frac{N_2}{q_2}.$$
Next, we have
$$
2^{(\beta_2-\beta_1) n}p_1=p_2+(x_12^{(\beta_2-\beta_1) n}-x_3)+(x_2-x_4)2^{(\beta_2+\gamma) n}
=
p_2+y_0+z_02^{(\beta_2+\gamma) n}.
$$
Therefore, we can calculate $p_1$ and $q_1=\frac{N_1}{p_1}$.
This terminates the proof.\qed
\end{proof}

\section{Experimental Results}\label{exp}

We provide some experiments to verify Assumption \ref{assu::solution} and the correctness of our analysis. We provide an efficient open source implementation of our algorithm for identifying ideal lattices in SageMath. The source code is available at: \textcolor{blue}{\url{https://github.com/fffmath/gifp}}.

The experiments were run on a computer configured with AMD Ryzen 5 2500U with Radeon Vega Mobile Gfx (2.00 GHz). We selected the parameter $n=\log(N)$ using gradients, validated our theory starting from small-scale experiments, and continually increased the scale of our experiments. The results are presented in Table \ref{ex1}:
\begin{table*}
\centering
\begin{tabular}{cccccccccc}
   \toprule
   n&$\alpha n$&$\beta n$&$\beta_1 n$&$\beta_2 n$&$\gamma n$&$m$&$\dim(\mathcal{L})$&$ \text{Time for LLL(s)}$&$  \text{Time for Gr\"{o}bner Basis(s)} $\\
   \midrule
    200&20&40&20&30&140&6&28&1.8620&0.0033\\
    500&50&100&50&75&350&6&28&3.1158&0.0043\\
    500&50&150&50&75&300&6&28&4.23898&0.0048\\
    1000&100&200&100&150&700&6&28&8.2277&0.0147\\
   \bottomrule
\end{tabular}
\caption{Some experimental results for the GIFP.}\label{ex1}
\end{table*}

As can be seen from Table \ref{ex1}, we chose various values of $n$, $\alpha n$, $\beta n$, $\beta_1 n$, $\beta_2 n$ and $\gamma n$ to investigate the behavior of our proposed algorithm. For each set of parameters, we recorded the time taken by the LLL algorithm and Gr\"{o}bner basis algorithm to solve the Generalized Integer Factorization Problem (GIFP).

It is important to note that our paper introduces a new variable, 'w', to eliminate some 'z'. Introducing multiple variables may intuitively make it more challenging to satisfy Assumption~\ref{assu::solution}. However, in practice, it is not necessary to satisfy Assumption~\ref{assu::solution} to find the desired 'p' and 'q'. 

For example, we usually yield $yz-C=0$ for some constant $C$. Then we can calculate $z_0$ by $z_0=q_2=\gcd(y_0z_0, N_2)=\gcd(C, N_2)$.

At the same time, if we abandon the introduction of 'w', the corresponding bound changes from $\gamma>4a(1-\sqrt{\alpha})$ to $\gamma>2a(2-\sqrt{\alpha})$. Even with only three variables in this case, we can still find 'p' and 'q' without satisfying Assumption~\ref{assu::solution}.

As the size of the problem increases, the computation time for LLL and Gr\"{o}bner basis algorithms also increases. Nevertheless, our algorithm's time complexity grows moderately compared to the problem size. Therefore, we can conclude that our algorithm is suitable for practical applications in the Generalized Integer Factorization Problem (GIFP).

Besides the Generalized Implicit Factoring Problem, we also conducted experiments on a special case, called the \textit{least-most significant bits case} (LMSBs). This case is characterized by $\beta_1=0$ and $\beta_2=\beta$. The results of these experiments are outlined below.
\begin{table*}
\centering
\begin{tabular}{cccccccc}
   \toprule
   n&$\alpha n$&$\beta n$&$\gamma n$&$m$&$\dim(\mathcal{L})$&$ \text{Time for LLL(s)}$&$  \text{Time for Gr\"{o}bner Basis(s)} $\\
   \midrule
256&25&75&156&5&21&1.3068&0.0029\\
256&25&75&156&5&21&1.2325&0.0023\\
256&25&75&156&6&21&1.2931&0.0023\\
512&50&150&212&6&28&2.0612&0.0028\\
512&50&150&212&6&28&2.4889&0.0086\\
512&50&150&212&6&28&2.0193&0.0022\\
   \bottomrule
\end{tabular}
\caption{Some experimental results for the LMSBs case.}\label{ex2}
\end{table*}

\section{Conclusion and Open Problem}
In this paper, we considered the Generalized Implicit Factoring Problem (GIFP), where the shared bits are not necessarily required to be located at the same positions. We proposed a lattice-based algorithm that can efficiently factor two RSA moduli, $N_1=p_1q_1$ and $N_2=p_2q_2$, in polynomial time, when the primes share a sufficient number of bits.

Our analysis shows that if $p_1$ and $p_2$ share $\gamma n > 4\alpha\left(1-\sqrt{\alpha}\right)n$ consecutive bits, not necessarily at the same positions, then $N_1$ and $N_2$ can be factored in polynomial time. However, this bound is valid when $p_i$ and $q_i$, $i=1,2$, are not assumed to have the same bit length, i.e., $N_1$ and $N_2$ are unbalanced moduli \cite{DBLP:journals/iacr/NitajA14}.

So our work raises an open question on improving the bound  $4\alpha\left(1-\sqrt{\alpha}\right)$, which would lead to better bounds for specific cases such as sharing some middle bits. It is known that the unshared bits in the Most Significant Bits (MSBs) or the Least Significant Bits (LSBs) are continuous, and only one variable is required when using variables to represent the unshared bits. This makes the MSBs or LSBs case easier to solve than the generalized case and achieves a better bound of $2\alpha\left(1-\alpha\right)$. However, the bound of the MSBs is not linear with the bound of the GIFP, which is unnatural. We hope that the gap between the bounds of the MSBs or LSBs and the GIFP case can be reduced.

\section*{Acknowledgement}
The authors would like to thank the reviewers of SAC 2023 for their helpful comments and suggestions. Yansong Feng and Yanbin Pan were supported in part by National Key Research and Development Project (No. 2018YFA0704705), National Natural Science Foundation of China (No. 62032009, 12226006) and
Innovation Program for Quantum Science and Technology under Grant 2021ZD0302902.

\nocite{*}
\bibliographystyle{alpha}
\bibliography{mybibliography}
\appendix
\section{Details of calculations in Section 3.2} \label{apdproof}
In this appendix, we present the details of calculations for the quantities $e_X$, $e_Y$, $e_Z$, $e_W$, $e_N$, and $e_M$ used in Section 3.2. We begin by a lemma that will be easily proven by induction. This lemma is well-known and can be found in many textbooks and references on combinatorics and discrete mathematics, such as Table 174 on page 174 in \cite{10.5555/562056}.

\begin{lemma}
    The equation $\sum_{i=0}^n\binom{i}{2}=\binom{n+1}{3}$ holds for any integer $n$.
\end{lemma}

Moving on, we provide the calculations for $e_X$ as:
\begin{align*}
e_X&=\sum_{i=0}^m\sum_{j=0}^{m-i}j
=\sum_{i=0}^m\binom{m-i+1}{2}
=\sum_{i=0}^m\binom{i+1}{2}\\
&=\binom{m+2}{3}
=\frac{1}{6}m(m+1)(m+2).
\end{align*}

The calculation of $e_Y$ is the same as $e_X$.

Next, we provide the calculation for $e_Z$:
\begin{align*}
e_Z&=\sum_{i=0}^m\sum_{j=0}^{m-i}(i+j-\min(s,i+j))\\
&=\sum_{i=0}^m\sum_{j=0}^{m-i}\max\{i+j-s,0\}\\
&=\sum_{t=s+1}^m\sum_{j=0}^{t}(t-s)\quad\text{(Let }t=i+j\text{)}\\
&=\sum_{t=s+1}^{m}(t-s)(t+1)\\
&=\sum_{t=0}^{m}(t-s)(t+1)-\sum_{t=0}^{s}(t-s)(t+1)\\
&=\sum_{t=0}^{m}t(t+1)-\sum_{t=0}^{m}s(t+1)-\sum_{t=0}^{s}t(t+1)+\sum_{t=0}^{s}s(t+1)\\
&=2\sum_{t=0}^{m}\binom{t+1}{2}-s\sum_{t=0}^{m}(t+1)-2\sum_{t=0}^{s}\binom{t+1}{2}+s\sum_{t=0}^{s}(t+1)\\
&=2\binom{m+2}{3}-s\binom{m+2}{2}+\frac{1}{6}\binom{s+2}{3}\\
&=\frac{1}{3}m(m+1)(m+2)+\frac{1}{6}s(s+1)(s+2)-\frac{1}{2}s(m+1)(m+2).
\end{align*}

Then, we provide the calculation for $e_W$:
\begin{align*}
e_W&=\sum_{i=0}^m\sum_{j=0}^{m-i}(s-\min(s,i+j))\\
&=\sum_{i=0}^s\sum_{j=0}^{s-i}(s-i-j)\\
&=\sum_{i=0}^s\sum_{j=0}^{s-i}s-\sum_{i=0}^s\sum_{j=0}^{s-i}i-\sum_{i=0}^s\sum_{j=0}^{s-i}j\\
&=\frac{1}{2}s(s+1)(s+2)-\frac{1}{6}s(s+1)(s+2)-\frac{1}{6}s(s+1)(s+2)\\
&=\frac{1}{6}s(s+1)(s+2).
\end{align*}
Furthermore, we provide the calculation for $e_N$:
\begin{footnotesize}
\begin{align*}
e_N&=\sum_{i=0}^t\sum_{j=0}^{m-i}(t-i)=\sum_{i=0}^t(t-i)(m-i+1)=\sum_{i=0}^t(t-i)(m+2-i-1)\\
&=(m+2)\sum_{i=0}^t(t-i)-\sum_{i=0}^t(t-i)(i+1)=(m+2)\binom{t+1}{2}-\sum_{i=0}^{t}t(i+1)+\sum_{i=0}^t i(i+1)\\
&=(m+2)\binom{t+1}{2}-t\binom{t+2}{2}+\sum_{i=0}^{t}2\binom{i+1}{2}=(m+2)\binom{t+1}{2}-t\binom{t+2}{2}+2\binom{t+2}{3}\\
&=\frac{1}{6}t(t+1)(3m-t+4).
\end{align*}
\end{footnotesize}
Finally, we provide the calculation for $e_M$:
\begin{align*}
e_M&=\sum_{i=0}^m\sum_{j=0}^{m-i}(m-i)=\sum_{i=0}^m(m-i+1)(m-i)=\sum_{i=0}^{m}2\binom{m-i+1}{2}\\&=\sum_{i=0}^{m}2\binom{i+1}{2}=2\binom{m+2}{3}=\frac{1}{3}m(m+1)(m+2).
\end{align*}
\end{document}